\newif\ifArxiv
\Arxivtrue

\newif\ifShowTodo
\ShowTodofalse

\newif\ifOmitnot
\Omitnottrue

\ifArxiv
\documentclass[11pt,a4paper]{article}
\else
\documentclass{sig-alternate-05-2015}
\fi

\usepackage[english]{babel}  
\usepackage[utf8]{inputenc}
\usepackage{graphicx}
\usepackage{amssymb}
\usepackage{amsmath}
\usepackage{xspace}          
\usepackage{varioref}        
\usepackage{subfigure}
\usepackage{booktabs}
\usepackage{multirow}        
\usepackage{color}
\usepackage{bm}

\usepackage{algorithm}
\usepackage{algpseudocode}
\usepackage[final]{microtype}
\usepackage{hyperref}
\usepackage[sort,nocompress]{cite}

\ifArxiv
\usepackage[amsmath,hyperref,thmmarks]{ntheorem}
\else
\fi
\newtheorem{theorem}{Theorem}
\newtheorem{lemma}[theorem]{Lemma}
\newtheorem{corollary}[theorem]{Corollary}

\newtheorem{definition}[theorem]{Definition}

\newtheorem{problem}[theorem]{Problem}
\newtheorem{example}[theorem]{Example}
\ifArxiv
\theoremstyle{nonumberplain}
\theoremsymbol{\ensuremath{\square}}
\theorembodyfont{\normalfont\upshape}
\newtheorem{proof}{Proof}
\else
\fi

\usepackage{fancyref}

\makeatletter
\def\mkfancyprefix#1#2{%
\expandafter\def\csname fancyref#1labelprefix\endcsname{#1}%
\begingroup\def\x{\endgroup\frefformat{plain}}%
    \expandafter\x\csname fancyref#1labelprefix\endcsname
    {\MakeLowercase{#2}\fancyrefdefaultspacing##1}%
\begingroup\def\x{\endgroup\Frefformat{plain}}%
    \expandafter\x\csname fancyref#1labelprefix\endcsname
    {#2\fancyrefdefaultspacing##1}%
\begingroup\def\x{\endgroup\frefformat{vario}}%
    \expandafter\x\csname fancyref#1labelprefix\endcsname
    {\MakeLowercase{#2}\fancyrefdefaultspacing##1##3}%
\begingroup\def\x{\endgroup\Frefformat{vario}}%
    \expandafter\x\csname fancyref#1labelprefix\endcsname
    {#2\fancyrefdefaultspacing##1##3}%
}
\makeatother
\fancyrefchangeprefix{\fancyrefeqlabelprefix}{eqn}
\Frefformat{vario}{\fancyrefeqlabelprefix}{(#1)#3}

\mkfancyprefix{ssec}{Section}
\mkfancyprefix{tbl}{Table}
\mkfancyprefix{thm}{Theorem}
\mkfancyprefix{lem}{Lemma}
\mkfancyprefix{cor}{Corollary}
\mkfancyprefix{prop}{Proposition}
\mkfancyprefix{prob}{Problem}
\mkfancyprefix{alg}{Algorithm}
\mkfancyprefix{inv}{Invariant}
\mkfancyprefix{ex}{Example}
\mkfancyprefix{line}{Line}
\mkfancyprefix{def}{Definition}
\mkfancyprefix{itm}{Item}
\mkfancyprefix{app}{Appendix}
\newcommand{\cref}[1]{\Fref{#1}}

\algrenewcommand\alglinenumber[1]{{\scriptsize#1}}   
\algrenewcommand\algorithmicrequire{\textbf{Input:}} 
\algrenewcommand\algorithmicensure{\textbf{Output:}} 
\newcommand{\ass}{\leftarrow}
\newcommand{\return}{\textbf{return }}

\definecolor{orange}{rgb}{1,0.5,0}
\ifShowTodo
\newcommand{\todo}[1]{\textcolor{red}{[\textit{#1}]}\xspace}
\newcommand{\arne}[1]{\textcolor{blue}{[\textit{#1}]}\xspace}
\newcommand{\jsrn}[1]{\textcolor{orange}{[\textit{#1}]}\xspace}
\else
\newcommand{\todo}[1]{}
\newcommand{\arne}[1]{}
\newcommand{\jsrn}[1]{}
\fi

\newcommand{\Pade}{Pad\'e\xspace}                
\ifArxiv
\newcommand{\PADE}{\Pade}
\else
\newcommand{\PADE}{PAD\'E\xspace}
\fi

\newcommand{\hlineSpace}[1]{%
  \hline%
  \\[\dimexpr-\normalbaselineskip+#1]%
  }

\newcommand{\word}[1]{\textnormal{#1}}

\newcommand{\mo}{{-1}}                   
\newcommand{\Transp}{^\top}
\newcommand{\adj}{\word{adj}}

\renewcommand\vec[1]{\bm{#1}}            

\newcommand{\floor}[1]{\lfloor #1 \rfloor}
\newcommand{\ceil}[1]{\lceil #1 \rceil}

\newcommand{\diagg}[1][\null]{%
  \ifx\null#1%
  \Gamma_{\vec g}%
  \else%
  \Gamma_{\vec g_{#1}}%
  \fi}

\newcommand\FF[1]{\mathbb F_{#1}\xspace}
\newcommand\ZZ{\mathbb Z\xspace}

\newcommand{\K}{{\mathsf{K}}}

\newcommand{\rem}{\word{rem}}
\newcommand{\row}{\word{row}}
\newcommand{\rowdeg}{\word{rowdeg}}

\newcommand{\Row}{\word{Row}}

\newcommand\Osoft{O^{\scriptscriptstyle \sim}\!}

\newcommand{\M}{{\mathsf{M}}}
\DeclareMathOperator{\loglog}{loglog}

\newcommand{\algoname}[1]{\ensuremath{\mathsf{#1}}\xspace}

\newcommand{\MinBasis}{\algoname{MinBasis}}
\newcommand{\PopovBasis}{\algoname{PopovBasis}}
\newcommand{\NegMinBasis}{\algoname{NegMinBasis}}

\ifArxiv
\else
\setcopyright{acmlicensed}
\conferenceinfo{ISSAC '16,}{July 19 - 22, 2016, Waterloo, ON, Canada}
\isbn{978-1-4503-4380-0/16/07}\acmPrice{\$15.00}
\doi{http://dx.doi.org/10.1145/2930889.2930933}
\fi

\title{Algorithms for Simultaneous \Pade Approximations%
  \ifArxiv%
  \footnote{
    $\copyright$ Johan Rosenkilde, Arne Storjohann. This is the authors'
    version of the work. It is posted here for your personal use. Not for
    redistribution. The definitive version was published
    in ISSAC '16, http://dx.doi.org/10.1145/2930889.2930933.
  }
}
\ifArxiv%
  \author{%
    Johan Rosenkilde, né Nielsen\\
    \small Technical University of Denmark\\
    \small Denmark\\
    \small jsrn@jsrn.dk
    \and
    Arne Storjohann \\
    \small University of Waterloo \\
    \small Canada \\
    \small astorjoh@uwaterloo.ca
    }
    \date{}
\else
\numberofauthors{2}
\author{
\alignauthor
Johan Rosenkilde, né Nielsen \\
       \affaddr{Technical University of Denmark}\\
       \affaddr{Denmark}\\
       \email{jsrn@jsrn.dk}
\alignauthor
Arne Storjohann \\
       \affaddr{University of Waterloo}\\
       \affaddr{Canada}\\
       \email{astorjoh@uwaterloo.ca}
}
\fi

\begin{document}

\maketitle

\begin{abstract}
  We describe how to solve simultaneous \Pade approximations over a power series ring $\K[[x]]$ for a field $\K$ using $\Osoft(n^{\omega - 1} d)$ operations in $\K$, where $d$ is the sought precision and $n$ is the number of power series to approximate.
  We develop two algorithms using different approaches.
Both algorithms return a reduced sub-bases that generates the complete
set of solutions to the input approximations problem that satisfy 
the given degree constraints.
  Our results are made possible by recent breakthroughs in fast computations of minimal approximant bases and Hermite \Pade approximations.
\end{abstract}

\section{Introduction}
\label{sec:intro}

\todo{Be consistent about ``Simultaneous'' or ``simultaneous''}

The Simultaneous \Pade approximation problem concerns approximating
several power series $S_1,\ldots,S_n \in \K[[x]]$ with rational
functions $\frac {\sigma_1}\lambda,\ldots,\frac {\sigma_n} \lambda$,
all sharing the same denominator $\lambda$.  In other words, for
some $d \in \ZZ_{\geq 0}$, we seek $\lambda \in \K[x]$ of low degree such
that each of
\[
  \rem(\lambda S_1,\  x^d) , \rem(\lambda S_2,\  x^d),\ \ldots,\ \rem(\lambda S_n,\  x^d)
\]
has low degree.  The study of Simultaneous \Pade approximations
traces back to Hermite's proof of the transcendence of $e$
\cite{hermite_sur_1878}. Solving Simultaneous \Pade
approximations has numerous applications, such as in coding theory,
e.g.~\cite{feng_generalization_1991,schmidt_collaborative_2009};
or in distributed, reliable computation \cite{clement_pernet_high_2014}.
Many algorithms have been developed for this problem, see e.g.~\cite{beckermann_uniform_1992,olesh_vector_2006,sidorenko_linear_2011,nielsen_generalised_2013} as well as the references therein.
Usually one cares about the regime where $d \gg n$.
Obtaining $O(n d^2)$ is classical through successive cancellation, see \cite{beckermann_uniform_1994} or \cite{feng_generalization_1991} for a Berlekamp--Massey-type variant.
Using fast arithmetic, the previous best was $\Osoft(n^\omega d)$, where $\omega$ is the exponent for matrix multiplication, see \cref{ssec:cost}.
That can be done by computing a minimal approximant basis with e.g.~\cite{giorgi_complexity_2003,GuptaSarkarStorjohannValeriote11}; this approach traces back to \cite{barel_general_1992,beckermann_uniform_1992}.
Another possibility which achieves the same complexity is fast algorithms for solving structured linear systems, e.g.~\cite{bostan_solving_2008}; see \cite{chowdhury_faster_2015} for a discussion of this approach.

\jsrn{Should we make these considerations in this paper, or in a possible journal version?
\cite{olesh_vector_2006} presents an algorithm which is faster when $N_i$ all equal some $N$, and if $d > N + $
Note that for generic input, there will then be no such solutions.
The algorithm has complexity $\Osoft(n k^{\omega-1} d)$, where $k$
}
\arne{Maybe we can add something in the conclusions about this, e.g., 
For the special case of  Problem 1 when all the $N_i$ are equal to some $N$, and all the $g_i$ are equal, 
and $d > N+N/k$ for some $k \in \\ZZ_{>0}$, 
it can be shown that the solution basis has dimension bounded by $k$, and
the algorithm of \cite{olesh_vector_2006} achieves a running time of $\Osoft(n k^{\omega-1} d)$ provided $d > N+N/k$.
It should be possible to extend the algorithms in this paper to blah blah.  We will present
this is in a future paper.}

A common description is to require $\deg \lambda < N_0$ for some degree bound $N_0$, and similarly $\deg \rem(\lambda S_1,\, x^d) < N_i$ for $i = 1,\ldots,n$.
The degree bounds could arise naturally from the application, or could be set such that a solution must exist.
A natural generalisation is also to replace the $x^d$ moduli with arbitrary $g_1,\ldots,g_n \in \K[x]$.
Formally, for any field $\K$:
\begin{problem}
  \label{prob:sim_pade}
  Given a tuple $(\vec S, \vec g, \vec N)$ where 
\begin{itemize}
\item
$\vec S = (S_1,\ldots,S_n) \in \K[x]^n$ is a sequence of polynomials,
\item $\vec g = (g_1,\ldots,g_n) \in \K[x]^n$ is a sequence of moduli polynomials with $\deg S_i < \deg g_i$ for $i=1,\ldots,n$,
\item and $\vec N = (N_0,\ldots,N_n)
 \in \ZZ_{\geq 0}^{n+1}$ are degree bounds
satisfying $1\leq N_0 \leq \max_i \deg g_i$ and $N_i \leq \deg g_i$ for $i=1,\ldots,n$, 
\end{itemize}
find, if it exists, a non-zero vector $(\lambda, \phi_1, \ldots, \phi_n)$ such that
\begin{enumerate}
\item $\lambda S_i \equiv \phi_i \mod g_i$ for $i = 1,\ldots, n$, and \label{p1item1}
\item $\deg \lambda < N_0$ and $\deg \phi_i < N_i$ for  $i=1,\ldots,n$.
\end{enumerate}
\end{problem}
We will call any vector $(\lambda, \phi_1, \ldots, \phi_n)$ as above \emph{a solution} to a given Simultaneous \Pade approximation problem.
Note that if the $N_i$ are set too low, then it might be the case that no solution exists.
\begin{example}
  \label{ex:simpade}
  Consider over $\FF 2[x]$ that $g_1 = g_2 = g_3 = x^5$, and 
    $\vec S = (S_1,S_2,S_3) = 
      \left(x^{4} + x^{2} + 1,\,x^{4} + 1,\,x^{4} + x^{3} + 1\right)$,
  with degree bounds $\vec N = (5, 3, 4, 5)$.
  Then $\lambda_1 = x^4 + 1$ is a solution, since $\deg \lambda_1 < 5$ and
  \[
    \lambda_1 \vec S \equiv
      \left(x^{2} + 1,\ 1,\ x^{3} + 1\right)
      \mod x^5 \ .
  \]
  $\lambda_2 = x^{3} + x$ is another solution, since
  \[
    \lambda_2 \vec S \equiv
      \left(x,\ x^{3} + x,\ x^4+x^3 + x\right)
      \mod x^5 \ .
  \]
  These two solutions are linearly independent over $\FF 2[x]$ and span all solutions.
\end{example}
Several previous algorithms for solving \cref{prob:sim_pade} are
more ambitious and produce an entire \emph{basis} of solutions
that satisfy the first output condition $\lambda S_i \equiv \phi_i \mod g_i$
for $i=1,\ldots,n$,
including solutions that do not satisfy the degree bounds stipulated
by the second output condition.  Our algorithms are slightly more
restricted in that we  only return the sub-basis that generates
the set of solutions that satisfy both output requirements of
\cref{prob:sim_pade}.  
Formally:
\begin{problem}
  \label{prob:sim_pade_basis}
  Given an instance of \cref{prob:sim_pade}, find a matrix $A \in \K[x]^{* \times (n+1)}$ such that:
    \begin{itemize}
      \item Each row of $A$ is a solution to the instance.
      \item All solutions are in the $\K[x]$-row space of $A$.
      \item $A$ is $(-\vec N)$-row reduced\footnote{%
        The notions $(-\vec N)$-degree, $\deg_{-\vec N}$ and $(-\vec N)$-row reduced are recalled in \cref{sec:preliminaries}.}.
    \end{itemize}
\end{problem}
The last condition ensures that $A$ is minimal, in a sense, according to the degree bounds $\vec N$, and that we can easily parametrise which linear combinations of 
the rows of $A$ are solutions.
We recall the relevant definitions and lemmas in \cref{sec:preliminaries}.

We will call such a matrix $A$ a \emph{solution basis}.  
In the complexities we report here, we cannot afford to compute 
$A$ explicitly.  For example, if all $g_i = x^d$,
the number of field elements required to explicitly
write down all of the entries of $A$ could be $\Omega(n^2d)$.
Instead, we remark that $A$ is completely given by
the problem instance as well as the first column of $A$, containing
the $\lambda$ polynomials.\footnote{%
  The restriction $N_i \leq \deg g_i$ in \cref{prob:sim_pade} ensures
  that for a given $\lambda$, the only possibilities for the $\phi_i$
  in a solution are $\rem(\lambda S_i, \ g_i)$.  In particular, if
  we allowed $N_i > \deg g_i$ then $(0,\ldots, 0, g_i, 0, \ldots,
  0)$ would be a solution which can not be directly reconstructed
  from its first element.
}
Our algorithms will therefore represent $A$ row-wise using the
following compact representation.  
\begin{definition}
For a given instance of \cref{prob:sim_pade_basis}, a \emph{solution
specification} is a tuple $(\vec \lambda,\vec \delta) \in
\K[x]^{k \times 1} \times \ZZ_{<0}^k$ such that the \emph{completion} of $\vec
\lambda$ is a solution basis, and where $\vec \delta$ are the $(-\vec N)$-degrees of the
rows of $A$.

The \emph{completion} of $\vec \lambda = (\lambda_1,\ldots,\lambda_k)\Transp$ is
the matrix
\[
  \begin{bmatrix}
    \lambda_1 & \rem(\lambda_1 S_1,\ g_1) & \ldots & \rem(\lambda_1 S_n,\ g_n) \\
    \vdots & & \ddots & \vdots \\
    \lambda_k & \rem(\lambda_k S_1,\ g_1) & \ldots & \rem(\lambda_k S_n,\ g_n) \\
  \end{bmatrix}
  \ .
\]
\end{definition}
Note that $\vec \delta$ will consist of only negative numbers, since any solution $\vec v$ by definition has $\deg_{-\vec N} \vec v < 0$.

\begin{example}
  A solution specification for the problem in \cref{ex:simpade} is
  \[
    (\vec \lambda, \vec \delta) = \big( [x^4 + 1,\ x^3 + x]\Transp ,\ (-1, -1) \big)  \ .
  \]
  The completion of this is
  \[
    A = \begin{bmatrix}
           x^4 + 1 & x^{2} + 1 & 1 & x^{3} + 1 \\
           x^3 + x & x  & x^{3} + x & x^4+x^3 + x
         \end{bmatrix}
  \]
  One can verify that $A$ is $(-\vec N)$-row reduced.
\end{example}

We present two algorithms for solving \cref{prob:sim_pade_basis},
both with complexity $O\big(n^{\omega-1}\, \M(d)\,(\log d)\,(\log
d/n)^2\big)$, where $d = \max_i \deg g_i$ and $\M(d)$ is the cost of multiplying two polynomials of degree $d$, see \cref{ssec:cost}.
They both depend crucially on recent developments
that allow computing minimal approximant bases of non-square matrices
faster than for the square case
\cite{zhou_efficient_2012,jeannerod_computation_2016}.
We remark that from the solution basis, one can also compute the expanded form of
one or a few of the solutions in the same complexity, for instance if a single, expanded solution to the simultaneous \Pade problem is needed.

Our first algorithm in \cref{sec:dual} assumes $g_i = x^d$ for all
$i$ and some $d \in \ZZ_{\geq 0}$.  It utilises a well-known duality between
Simultaneous \Pade approximations and Hermite \Pade approximations,
see e.g.~\cite{beckermann_uniform_1992}.  The Hermite \Pade problem
is immediately solvable by fast minimal approximant basis computation.
A remaining step is to efficiently compute a single row of the
adjoint of a matrix in Popov form, and this is done by combining
partial linearisation \cite{GuptaSarkarStorjohannValeriote11} and high-order
lifting \cite{storjohann_high-order_2003}.

Our second algorithm in \cref{sec:intersect} supports arbitrary $g_i$.
The algorithm first solves $n$ single-sequence \Pade approximations, each of $S_1,\ldots,S_n$.
The solution bases for two problem instances can be combined by computing the
intersection of their row spaces; this is handled by a minimal approximant basis
computation.
A solution basis of the full Simultaneous \Pade problem is then obtained by structuring intersections along a binary tree.

Before we describe our algorithms, we give some preliminary notation and definitions in \cref{sec:preliminaries}, and in \cref{sec:subroutines} we describe some of the computational tools that we employ.

Both our algorithms have been implemented in Sage v. 7.0 \cite{stein_sagemath_????} (though asymptotically slower alternatives to the computational tools are used).
The source code can be downloaded from \url{http://jsrn.dk/code-for-articles}.

\subsection{Cost model}
\label{ssec:cost}

We count basic arithmetic operations in $\K$ on an algebraic RAM.
We will state complexity results in terms of an exponent $\omega$
for matrix multiplication, and a function $\M(\cdot)$ that is a
 multiplication time for
$\K[x]$ \cite[Definition~8.26]{von_zur_gathen_modern_2012}.  Then two
$n\times n$ matrices over $\K$ can be multiplied in $O(n^{\omega})$
operations in $\K$, and two polynomials in $\K[x]$ of degree
strictly less than $d$ can be multiplied in $\M(d)$ operations in 
$\K$.  The best known algorithms allow $\omega < 2.38$
\cite{coppersmith_matrix_1990, LeGall14}, and we can always take
$\M(d) \in O(n (\log n) (\loglog n))$ \cite{CantorKaltofen}.  

In
this paper we assume that $\omega > 2$, and that $\M(d)$ is super-linear while
$\M(d) \in O(d^{\omega-1})$.  The assumption $\M(d) \in O(d^{\omega-1})$ simply
stipulates that if fast matrix multiplication techiques are used
then fast polynomial multiplication should be used also:
for example, $n \, \M(nd) \in O(n^{\omega} \, \M(d))$.

\section{Preliminaries}
\label{sec:preliminaries}
Here we gather together some definitions and results regarding row
reduced bases, minimal approximant basis, and their shifted variants.
For a matrix $A$ we denote by $A_{i,j}$ the entry in row $i$ and
column $j$.  For a matrix $A$ over $\K[x]$ we denote by $\Row(A)$
the $\K[x]$-linear row space of $A$.



\subsection{Degrees and shifted degrees}

The degree of a nonzero vector $\vec v \in \K[x]^{1 \times m}$ or
matrix $A \in \K[x]^{n\times m}$ is denoted by
$\deg \vec v$ or $\deg A$, and is the maximal degree of entries of
$\vec v$ or $A$.  If $A$ has no zero rows the {\em row degrees} of $A$, denoted
by $\rowdeg\, A$, is the tuple $(d_1,\ldots,d_n)$ with $d_i = \deg
\row(A,i)$.

The (row-wise) {\em leading matrix} of $A$, denoted by 
${\rm LM}(A) \in \K^{n \times m}$, has ${\rm LM}(A)_{i,j}$ 
equal to the coefficient of $x^{d_i}$ of $A_{i,j}$.

Next we recall~\cite{barel_general_1992,zhou_efficient_2012,jeannerod_computation_2016}
the shifted variants of the notion of degree, row degrees, and leading
matrix.  For a {\em shift} $\vec s =
(s_1,\ldots,s_n) \in \ZZ^n$, 
define the $n \times n$ diagonal matrix $x^{\vec s}$ 
by $$x^{\vec s} := \left [ \begin{array}{ccc} x^{s_1}
& & \\
 &  \ddots & \\ & & x^{s_n} \end{array} \right ].$$
Then the {\em ${\vec s}$-degree} of $v$, the {\em ${\vec s}$-row
degrees} of $A$, and the {\em $\vec s$-leading matrix} of $A$, are
defined by \vec $\deg_{\vec s} v := \deg v x^{\vec s}$, $\rowdeg_{\vec
s} A := \rowdeg\, Ax^{\vec s}$, and ${\rm LM}_{\vec s}(A) := {\rm
LM}(Ax^{\vec s})$.  
Note that we pass over the ring of Laurent polynomials only for
convenience; our algorithms will only compute with polynomials.
As pointed out in~\cite{jeannerod_computation_2016}, up to negation
the definition of ${\vec s}$-degree is equivalent to that used
in~\cite{BeckermannLabahnVillard06} and to the notion of {\em defect}
in~\cite{beckermann_uniform_1994}.

For an instance $(\vec S, \vec g, \vec N)$ of \cref{prob:sim_pade}, in the
context of defining matrices, we will be using $\vec S$ and $\vec g$ as
vectors, and by $\diagg$ denote the diagonal matrix with the entries of
$\vec g$ on its diagonal.

\subsection{Row reduced}

Although row reducedness can be defined for matrices of arbitrary
shape and rank, it suffices here to consider the case of matrices
of full row rank.  A matrix $R \in \K[x]^{n \times m}$ is
{\em row reduced} if ${\rm LM}(R)$ has full row rank, and
{\em $\vec{s}$-row reduced} if ${\rm LM}_{\vec s}(R)$ has full row rank.
Every $A \in \K[x]^{n \times m}$ of full row rank is left equivalent
to a matrix $R \in \K[x]^{n \times m}$ that is ${\vec s}$-row
reduced.  The rows of $R$ give a basis for $\Row(A)$ that is minimal
in the following sense: the list of ${\vec s}$-degrees of the rows
of $R$, when sorted in non-decreasing order, will be lexicographically
minimal.  An important feature of row reduced matrices is 
the so-called ``predictable degree''-property~\cite[Theorem~6.3-13]{kailath_linear_1980}: for any 
$\vec v \in \K[x]^{1 \times n}$, we have
  \[
    \deg_{\vec s}(\vec v R) = \max_{i=1,\ldots,n}( \deg_{\vec s} {\rm row}(R,i)
+ \deg v_i ) \ .
  \]

A canonical $\vec{s}$-reduced basis is provided by the ${\vec
s}$-Popov form.  Although an ${\vec s}$-Popov form can be defined
for a matrix of arbitrary shape and rank, it suffices 
here to consider the case of a non-singular matrix.  The
following definition is equivalent
to~\cite[Definition~1.2]{jeannerod_computation_2016}.

\begin{definition}  \label{def:popov}
A non-singular matrix $R \in \K[x]^{n\times n}$ is in ${\vec s}$-Popov
form if ${\rm LM}_{\vec s}(R)$ is unit lower triangular and the
degrees of off-diagonal entries of $R$ are strictly less than the
degree of the diagonal entry in the same column.
\end{definition}


\subsection{Adjoints of row reduced matrices}

For a non-singular matrix $A$ recall that the adjoint of $A$, denoted
by ${\rm adj}(A)$, is equal to $(\det A)A^{-1}$, and that entry
${\adj(A)}_{i,j}\Transp$ is equal to $(-1)^{i+j}$ times the determinant
of the $(n-1) \times (n-1)$ sub-matrix that is obtained from $A$ by
deleting row $i$ and column $j$.


\begin{lemma}
  \label{lem:adjointRowReduced}
  Let $A \in \K[x]^{n \times n}$ be $\vec s$-row reduced.  Then 
$\adj(A)\Transp$ is $(-\vec s)$-row reduced with
\[
  \rowdeg_{(-\vec s)} \adj(A)\Transp =(\eta - s - \eta_1,\ldots , \eta - s -\eta_n) \ ,
\]
where $\vec \eta = \rowdeg_{\vec s} A$, $\eta = \sum_i \eta_i$ and $s = \sum_i s_i$.
\end{lemma}
\begin{proof}
Since $A$ is $\vec s$-row reduced then $A x^{\vec s}$ is row reduced.
Note that $\adj(A x^{\vec s})\Transp (A x^{\vec s})\Transp  
= (\det A x^{\vec s}) I_m$ with $\deg
\det A x^{\vec s} = \eta$.  It follows that 
row $i$ of $\adj(A x^{\vec s})\Transp$ must
have degree at least $\eta - \eta_i$ since
$\eta_i$ is the degree of column $i$ 
of $(A x^{\vec s})\Transp$. However, entries in row 
$i$ of $\adj(A x^{\vec s})\Transp$ are minors of the matrix obtained from
$A x^{\vec s}$ by removing row $i$, hence have degree at most $\eta
- \eta_i$.  It follows that the (row-wise) leading coefficient matrix of
$\adj(A x^{\vec s})\Transp$ is non-singular, hence $\adj (A x^{\vec
s})\Transp$ is row reduced.   Since $\adj (A x^{\vec s})\Transp =  
(\det x^{\vec s}) \adj(A)\Transp x^{-\vec s}$ we conclude that $\adj(A)\Transp$
is $(-\vec s)$-row reduced with $\rowdeg_{(-\vec s)} \adj(A) =
(\eta - \eta_1 - s, \ldots, \eta - \eta_n - s)$.
\end{proof}


\subsection{Minimal approximant bases}

We recall the standard notion of minimal approximant basis, sometimes
known as order basis or $\sigma$-basis \cite{beckermann_uniform_1994}.
For a matrix $A \in \K[x]^{n \times m}$ and order $d \in
\ZZ_{\geq 0}$, an \emph{order $d$ approximant} is a vector $\vec p \in
\K[x]^{1 \times n}$ such that
$\vec pA \equiv \vec 0 \mod x^d.$

An \emph{approximant basis of order $d$} is then a matrix $F \in
\K[x]^{n \times n}$ which is a basis of all order $d$ approximants.
Such a basis always exists and has full rank $n$.  
For a shift $\vec s \in \ZZ^n$,
$F$ is then an
\emph{$\vec s$-minimal approximant basis} if it is $\vec s$-row
reduced.

Let $\MinBasis(d,A,\vec s)$ be a function that returns $(F,\vec
\delta)$, where $F$ is an $\vec s$-minimal approximant basis of $A$
of order $d$, and $\vec \delta  = \rowdeg_{\vec s} F$.  The next
lemma recalls a well known method of constructing minimal approximant
bases recursively.  Although the output of $\MinBasis$ may not be
unique, the lemma holds for \emph{any} $\vec s$-minimal approximant basis 
that $\MinBasis$ might return.
\arne{Give some refs regarding lem. I was hoping to ask George, but he is away.
I'd rather not cite something incorrectly so let's just leave it.}
\jsrn{Can we make it more clear that all the following properties hold for
\emph{any} $\vec s$-minimal approximant basis that $\MinBasis(...)$ might return?}

\begin{lemma} \label{lem:paderec} Let $A = \left [ \begin{array}{c|c}
A_1 & A_2 \end{array} \right ]$ over $\K[x]$. If
$(F_1, \vec \delta_1) = \MinBasis(d,A_1,\vec s)$
and $(F_2,\vec \delta_2) =
\MinBasis(d,F_1A_2,\vec \delta_1)$, then
$F_2F_1$ is an $\vec s$-minimal approximant basis of $A$ of order $d$
with $\vec \delta_2 = \rowdeg_{\vec s} F_2 F_1$.
\end{lemma}

Sometimes only the {\em negative part} of an $\vec s$-minimal
approximant bases is required, the submatrix of the approximant
bases consisting of rows with negative $\vec s$-degree.
Let function $\NegMinBasis(d,A,\vec
s)$ have the same output as $\MinBasis$, but with
$F$ restricted to the negative part.
\arne{Give some refs for the cor.}

\begin{corollary} \label{lem:paderecprune} \cref{lem:paderec} still
holds if $\MinBasis$ is replaced by $\NegMinBasis$, and 
``an $\vec s$-minimal'' is replaced with ``the negative part of an $\vec s$-minimal.''
\end{corollary}

Using for example the algorithm \texttt{M-Basis} of
\cite{giorgi_complexity_2003}, it is easy to show 
that any order
$d$ approximant basis $G$ for an $A$ of column dimension $m$ has
$\det G = x^D$ for some $D \in \ZZ_{\geq 0}$ with $D \leq md$.

Many problems of $\K[x]$ matrices or approximations reduce to the
computation of (shifted) minimal approximant bases, see
e.g.~\cite{beckermann_uniform_1994,giorgi_complexity_2003},
often resulting in the best known asymptotic complexities for these
problems.

\subsection{Direct solving of Simultaneous \Pade approximations}
\label{sec:direct_solve}


Let $(\vec S, \vec g, \vec N)$ be an instance of \cref{prob:sim_pade_basis}
of size $n$.  We recall some \jsrn{should we add: ``but not all''?} known approaches for computing a solution
specification using row reduction and minimal approximant basis
computation.





\subsubsection{Via reduced basis}
\label{sec:direct_reduced_basis}

Using the predictable degree property 
it is easy to show that if $R \in \K[x]^{(n+1)
\times (n+1)}$ is an $(-\vec N)$-reduced basis of
\[ A = 
 \left [ \begin{array}{c|c} 
  1 & \vec S \\\hline
 & \diagg
         \end{array} \right]
       \in \K[x]^{(n+1) \times (n+1)},
    \]
then the sub-matrix 
of $R$ comprised of the rows with negative $(-\vec N)$-degree form
a solution basis.  A solution specification $(\vec \lambda, \vec
\delta)$ is then a subvector $\vec \lambda$ 
of the first column of $R$, with $\vec \delta$
the corresponding subtuple $\vec \delta$ of $\rowdeg_{(- \vec N)} R$.

Mulders and Storjohann \cite{mulders_lattice_2003} gave an iterative algorithm
for performing row reduction by successive cancellation; it is similar to but
faster than earlier algorithms
\cite{kailath_linear_1980,lenstra_factoring_1985}.
Generically on input $F \in \K[x]^{m \times m}$ it has complexity
$O(n^3 (\deg F)^2)$.
Alekhnovich \cite{alekhnovich_linear_2005} gave what is essentially a Divide \&
Conquer variant of Mulders and Storjohann's algorithm, with complexity
$\Osoft(n^{\omega+1}\deg F)$.
Nielsen remarked \cite{nielsen_generalised_2013} that these algorithms
perform fewer iterations when applied to the matrix $A$ above, due to its low
\emph{orthogonality defect}: ${\rm OD}(F) = \sum\rowdeg F - \deg \det F$, resulting
in $O(n^2(\deg A)^2)$ respectively $\Osoft(n^\omega \deg A)$.
Nielsen also used the special shape of $A$ to give a variant of the
Mulders--Storjohann algorithm that computes coefficients in the working matrix in a lazy
manner with a resulting complexity $O(n \,\mathsf{P}(\deg A))$, where
$\mathsf P(\deg A) = (\deg A)^2$ when the $g_i$ are all powers of $x$, and
$\mathsf P(\deg A) = \M(\deg A)\deg A$ otherwise.

Giorgi, et al. \cite{giorgi_complexity_2003} gave a reduction for performing
row reduction by computing a minimal approximant basis.
For the special matrix $A$, this essentially boils down to the approach
described in the following section.

When $n = 1$, the extended Euclidean algorithm on input $S_1$ and $g_1$ can solve the approximation problem by essentially computing the reduced basis of the $2 \times 2$ matrix $A$: each iteration corresponds to a reduced basis for a range of possible shifts \cite{sugiyama_further_1976,justesen_complexity_1976,gustavson_fast_1979}.
The complexity of this is $O(\M(\deg g_1) \log \deg g_1)$.

\subsubsection{Via minimal approximant basis}

First consider the special case when all $g_i = x^d$ for the same
$d$.  An approximant $\vec v = (\lambda, \phi_1, \ldots, \phi_n)$
of order $d$ of
\begin{align*}
   A &=  \left [ \begin{array}{c}
 - \vec S \\
 I
            \end{array} \right ]
       \in \K[x]^{(n+1) \times n}
\end{align*}
clearly satisfies $\lambda S_i \equiv \phi_i \mod x^d$ for $i =
1,\ldots,n$; conversely, any such vector $\vec v$ 
satisfying these congruences must be an approximant of $A$
of order $d$.  So the negative part of a $(-\vec N)$-minimal
approximant basis of $A$ of order $d$ is a solution basis.

In the general case we can reduce to a minimal approximant bases
computation as shown by \cref{alg:simpadedirect}. 
Correctness of the algorithm follows from the following result.
\begin{theorem} \label{thm:simpadedirect}
Corresponding to an instance $(\vec S, \vec g, \vec N)$ of 
\cref{prob:sim_pade_basis} of size $n$, define a shift
$\vec h$ and order $d$:
\begin{itemize}
\item $\vec h := -(\vec N \mid N_0 -1, \ldots, N_0-1) \in \ZZ^{2n+1}$
\item $d := N_0 + \max_i \deg g_i -1$
\end{itemize}
If $G$ is the negative part of an $\vec h$-minimal approximant basis of
$$
H = \left [ \begin{array}{c} -\vec S \\
 I \\
 \diagg \end{array} \right ] \in \K[x]^{(2n+1) \times n}
$$
of order $d$, then the submatrix of $G$ comprised of the 
first $n+1$ columns is a solution basis to the problem instance.
\end{theorem}
\begin{proof} An approximant 
$\vec v = (\lambda, \phi_1,\ldots,\phi_n,q_1,\ldots, q_n)$ of order
$d$ of $H$ clearly satisfies 
\begin{align} \label{eq:lambda} 
\lambda S_i = \phi_i + q_ig_i \bmod
x^d
\end{align} for $i=1,\ldots,n$; conversely, any such vector $\vec v$
satisfying these
congruences must be an approximant of $H$ of order $d$.

Now suppose $\vec v$ is an order $d$ approximant of $H$ with negative
$\vec h$-degree, so $\deg \lambda \leq N_0-1$, $\deg \phi_i \leq
N_i-1$, and $\deg q_i \leq  N_0 - 2$.  
Since \cref{prob:sim_pade}
specifies that $\deg S_i < \deg g_i$ and 
$N_i \leq \deg g_i$, both $\lambda S_i$ and $q_i g_i$ 
will have degree bounded by $N_0 + \deg g_i - 2$.  
Since \cref{prob:sim_pade} specifies that $N_0 \geq 1$,
it follows that both the left and right hand sides of (\ref{eq:lambda})
have degree bounded by $N_0+ \deg g_i -2$, which is strictly less
than $d$.  We conclude that 
\begin{align} \label{eq:lambda2}
\lambda S_i = \phi_i + q_i g_i
\end{align} for
$i=1,\ldots,n$.  
It follows that $\vec v H = 0$ so $\vec v$ is in the left kernel
of $H$.  Moreover, restricting $\vec v$ to its first $n+1$ entries
gives $\bar{\vec v} := (\lambda, \phi_1,\ldots,\phi_n)$, a solution
to the simultaneous \Pade problem with $\deg_{- \vec N} \bar{\vec
v} = \deg_{\vec h} \vec v$.
Conversely, if $\bar{\vec v} = (\lambda, \phi_1,\ldots,\phi_n)$ is
a solution to the simultaneous \Pade problem, then the extension
$\vec v = (\lambda, \phi_1,\ldots,\phi_n,q_1,\ldots,q_n)$ with $q_i
= (\lambda S_i - \phi_i)/g_i \in \K[x]$ for $i=1,\ldots,n$ is an
approximant of $H$ of order $d$ with $\deg_{\vec h} \vec v =
\deg_{-\vec N} \bar{\vec v}$.

Finally, consider that a left kernel basis for $H$ is given by
$$
K = \left[ \begin{array}{c|c} K_1 & K_2 \end{array} \right ] = 
\left [ \begin{array}{cc|c} 1 & \vec S & \\
 & \diagg & -I \end{array} \right ].
$$
We must have $G = M K$ for some polynomial matrix $M$ of full row
rank.  But then $M K_1$ also has full row rank with $\rowdeg_{-\vec N} MK_1
= \rowdeg_{\vec h} G$.
\end{proof}
\begin{algorithm}[t]
  \caption{\algoname{DirectSimPade}}
  \label{alg:simpadedirect}
  \begin{algorithmic}[1]
   \Require{$(\vec S, \vec g, \vec N)$, an instance of
   \cref{prob:sim_pade_basis} of size $n$.}
    \Ensure{$(\vec \lambda,  \vec \delta)$, a solution specification.}
    \State $\vec h \ass -( \vec N \mid  N_0-1,\ldots,N_0-1) \in \ZZ^{2n+1}$
    \State $d \ass N_0 + \max_i \deg g_i - 1$
    \State $H = 
  \left[ \begin{array}{c}
       -\vec S \\
  I \\
  \diagg
         \end{array} \right]$
    \State $(\left [ \begin{array}{c|c} \vec \lambda & \ast \end{array} 
  \right ], \vec \delta) \ass \NegMinBasis(d, H, \vec h)$
    \State $\return (\vec \lambda, \vec \delta)$
  \end{algorithmic}
\end{algorithm}

\algoname{DirectSimPade} can be performed in time 
$\Osoft(n^{\omega} \deg H) = \Osoft(n^\omega \max_i \deg g_i)$ using the minimal approximant basis algorithm by Jeannerod, et al.~\cite{jeannerod_computation_2016}, see \cref{sec:subroutines}.

A closely related alternative to \algoname{DirectSimPade} is the recent algorithm by Neiger \cite{neiger_fast_2016} for computing solutions to modular equations with general moduli $g_i$.
This would give the complexity $\Osoft(n^{\omega-1} \sum_i \deg g_i) \subset \Osoft(n^\omega \max_i \deg g_i)$.

All of the above solutions ignore the sparse, simple structure of the input
matrices, which is why they do not obtain the improved complexity that we do here.

\section{Computational tools}
\label{sec:subroutines}

The main computational tool we will use is the following very recent
result from Jeannerod, Neiger, Schost and
Villard~\cite{jeannerod_computation_2016} on minimal approximant
basis computation.

\begin{theorem}[\protect{\cite[Special case of Theorem~1.4]{jeannerod_computation_2016}}]
  \label{thm:orderbasis}
  There exists an algorithm $\PopovBasis(d,A, \vec s)$ where the
  input is an order $d \in \ZZ_+$, a polynomial matrix $A \in \K[x]^{n
  \times m}$ of degree at most $d$, and shift $\vec s \in \ZZ^n$, 
and which returns $(F, \vec \delta)$, where
  $F$ is an $\vec s$-minimal approximant basis of $A$ of order $d$,
 $F$ is in $\vec s$-Popov form, and $\vec \delta = \rowdeg_{\vec s} F$. 
 \PopovBasis has complexity 
$O(n^{\omega-1}\, \M(\sigma)\, (\log \sigma) \, (\log \sigma /n)^2)$ operations in $\K$, where $\sigma = md$.
\end{theorem}
Our next result says that we can quickly compute the first row of $\adj(F)$ if $F$ is a minimal approximant basis in Popov form.
In particular, since $F$ is an approximant basis $\det F = x^D$ for some $D \leq \sigma$, where $\sigma = md$ from \cref{thm:orderbasis}.
\begin{theorem}   \label{thm:fastsolver}
  Let $F \in \K[x]^{n \times n}$ be in Popov form and with $\det F = x^D$ for some $D \in \ZZ_{\geq 0}$.
  Then the first row of $\adj(F)$ can be computed in $O(n^{\omega-1}\, \M(D)\, (\log D) \, (\log D/n))$ operations in $\K$.
\end{theorem}
\begin{proof}
Because $F$ is in $\vec s$-Popov form, $D$ is the sum of the column degrees of $F$.
We consider two cases: $D \geq n$ and $D < n$.

First suppose $D
\geq n$.  Partial linearisation 
\cite[Corollary~2]{GuptaSarkarStorjohannValeriote11}
can produce from $F$, with
no operations in $\K$, a new matrix $G \in \K[x]^{\bar n \times \bar n}$ with
dimension $\bar{n} < 2n$, $\deg G \leq \lceil D/n\rceil$, $\det G = \det F$, 
and such that $F^\mo$ is equal to the principal $n \times n$
sub-matrix of $G^\mo$.  Let $\vec v \in \K[x]^{1 \times \bar{n}}$ 
be the first row of $x^DI_{\bar{n}}$.
Then the first row of $\adj(F)$ will be the first $n$ entries of
the first row of $\vec vG^{-1}$.  High-order $X$-adic lifting
\cite[Algorithm~5]{storjohann_high-order_2003} using the modulus $X=(x-1)^{\lceil
D/n \rceil}$ will compute $\vec vG^{-1}$ in $O\big(n^{\omega}\,
\M(\lceil D/n \rceil) \,(\log \lceil D/n \rceil)\big)$ operations
in $\K$ \cite[Corollary~16]{storjohann_high-order_2003}.  Since $D \geq n$
this cost estimate remains valid if we replace $\lceil D/n \rceil$
with $D/n$.  Finally, from the super-linearity assumption on $\M(\cdot)$
we have $M(D/n) \leq (1/n) \M(D)$, thus matching our target cost.

Now suppose $D < n$. In this case we can not directly appeal to the
partial linearisation 
technique since the resulting
 $O(n^{\omega} \lceil D/n\rceil)$ may be asymptotically
larger than our target cost.  But $D < n$ means that $F$ has ---
possibly many --- columns of degree 0; since $F$ is in Popov form,
such columns have a 1 on the matrix's diagonal and are 0 on the
remaining entries.  The following describes how to essentially
ignore those columns.  $D$ is then greater than or equal to the
number of remaining columns, thus effectuating the gain from the
partial linearisation.

If $n-k$ is the number of such columns in $F$ that means we can find a
permutation matrix $P$ such that 
\[
  \hat{F} := PFP\Transp = \left [ \begin{array}{c|c}
                             F_1 & \\\hline
                             F_2 & I_{n-k}
                           \end{array} \right ] \ , 
\]
with each column of $F_1$ having degree strictly greater than zero.
Let $i$ be the row index of the single 1 in the first column of
$P\Transp$. Since $F^{-1} = P\Transp \hat{F}^{-1}P$, we have
\begin{equation}
\label{first}
{\rm row}(\adj(F),1)P^{-1} = x^D\, {\rm row}(\hat{F}^{-1},i).
\end{equation}
Considering that
\[
\hat{F}^{-1} = \left [ \begin{array}{c|c} F_1^{-1} & \\\hline -F_2F_1^{-1} & I_{n-k} \end{array} \right ],
\]
it will suffice to compute the first $k$ entries of the vector on
the right hand side of~(\ref{first}).  If $i \leq k$ then let $\vec
v \in \K[x]^{1 \times k}$ be row $i$ of $x^{D}I_k$. Otherwise, if
$i>k$ then let $\vec v$ be row $i-k$ of $-x^{D}F_2$. Then in both cases, $\vec
vF_1^{-1}$ will be equal to the first $k$ entries of the vector on
the right hand side of~(\ref{first}).  Like before, high-order
lifting combined with partial linearisation will compute this vector
in $O\big(k^{\omega}\, \M(\lceil D/k \rceil)\,(\log \lceil D/k \rceil)
\big)$ operations in $\K$. Since $D\geq k$ the cost estimate
remains valid if $\lceil D/k \rceil$ is replaced with $D/k$.
\end{proof}

\section{Reduction to Hermite \PADE}
\label{sec:dual}

In this section we present an algorithm for solving
\cref{prob:sim_pade_basis} when $g_1 = \ldots = g_n = x^d$ for some
$d \in \ZZ_{\geq 0}$.  The algorithm is based on the well-known duality
between the Simultaneous \Pade problem and the Hermite \Pade problem,
see for example~\cite{beckermann_uniform_1992}.  This duality, first
observed in a special case~\cite{Mahler68}, and then later in the
general case~\cite{beckermann_recursiveness_1997}, was exploited
in~\cite{beckermann_fraction-free_2009} to develop algorithms 
for the fraction free computation of Simultaneous \Pade approximation.  
We begin with a technical
lemma that is at the heart of this duality.

\begin{lemma}
  \label{lem:duality}
  Let $\hat A, \hat B \in \K[x]^{(n+1)\times(n+1)}$ be as follows.
  \begin{align*}
    \hat A &= 
         \left [ \begin{array}{c|c}
                   x^d & -\vec S \\\hline
                       & I
                 \end{array} \right ]
    &&
    \hspace*{-1em}\hat B &= 
           \left [ \begin{array}{c|cccc}
                     1 & \\\hlineSpace{2pt}
                     \vec S\Transp & x^d I
                   \end{array} \right ]
  \end{align*}
  Then $\hat B$ is the adjoint of $\hat A\Transp$.
  Furthermore, $\hat A\Transp$ is an approximant basis for $\hat
  B$ of order $d$, and $\hat B\Transp$ is an approximant basis of
  $\hat A$ of order $d$.
\end{lemma}
\begin{proof}
Direct computation shows that $\hat A\Transp \hat B = x^d I_m =
\det \hat A\Transp I_m$, so $\hat B$ is the adjoint of $\hat
A\Transp$.

Let now $G$ be an approximant basis of $\hat B$.  By the above
computation the row space of $\hat A\Transp$ must be a subset of
the row space of $G$.  But since $G \hat B = (x^dI_m) R$ for some
$R \in \K[x]^{(n+1)\times(n+1)}$, then $\det G = x^d \det R$.  Thus
$x^d \mid \det G$.  But $\det \hat A\Transp = x^d$, so the row space
of $\hat A\Transp$ can not be smaller than the row space of $G$.
That is, $\hat A\Transp$ is an approximant basis for $B$ of order
$d$.  Taking the transpose through the argument shows that $\hat
B\Transp$ is an approximant basis of $\hat B$ of order $d$.
\end{proof}
\begin{theorem}
  \label{thm:dualityMinbasis}
  Let $A$ and $B$ be as follows.
  \begin{align*}
    A &= \left [ \begin{array}{cccc}
                   -\vec S \\\hlineSpace{1pt}
                   I
                 \end{array} \right ] \in \K[x]^{(n+1) \times (n+1)}
    &&&
    \hspace*{-1em}B  &= \left [ \begin{array}{c}
                    1 \\ \vec S
                  \end{array} \right] \in \K[x]^{(n+1) \times 1}
\end{align*}
If $G$ is an $\vec N$-minimal approximant basis of $B$ of order $d$ with shift
$\vec N \in \ZZ_{\geq 0}^{n+1}$, then 
$\adj(G\Transp)$ is a $(-\vec N)$-minimal
approximant basis of $A$ of order $d$.  Moreover,
if $\vec \eta = \rowdeg_{\vec N} G$, then
$\rowdeg_{-\vec N} \adj(G) =(\eta - N - \eta_1,\ldots , \eta - N
-\eta_{n+1})$, where $\eta = \sum_i \eta_i$ and $N = \sum_i N_i$.
\end{theorem}

\begin{proof}
  Introduce $\hat A$ and $\hat B$ as in \cref{lem:duality}.
  Clearly $G$ is also an $\vec N$-minimal approximant basis of $\hat B$ of order $d$.
  Likewise, $\hat A$ and $A$ have the same minimal approximant bases for given order and shift.

  Assume, without loss of generality, that we have scaled $G$ such
  that $\det G$ is monic.  Since $\hat A\Transp$ is also an approximant
  basis for $\hat B$ of order $d$, then $\det G = \det \hat A\Transp
  = x^d$.  By definition $G\hat B = x^d R$ for some matrix $R \in
  \K[x]^{(n+1)\times(n+1)}$.  That means
  \begin{align*}
    x^{2d}((G\hat B)\Transp))^\mo &= x^{2d}((x^d R)\Transp)^\mo \ , & \textrm{so} \\ 
    (x^d(G\Transp)^\mo)(x^d(\hat B\Transp)^\mo) &= x^{d}(R\Transp)^\mo \ , & \textrm{that is} \\ 
    \adj(G\Transp) \hat A &= x^d (R\Transp)^\mo \ .
  \end{align*}
Now $\det R = 1$ since $(x^d)^{n+1} \det R = \det(G\hat B) =
x^{d+nd}$, so $(R\Transp)^\mo = \adj(R\Transp) \in \K[x]^{(n+1)
\times (n+1)}$.
Therefore $\adj(G\Transp)$ is an approximant basis of $\hat A$
of order $d$.   
The theorem now follows from \cref{lem:adjointRowReduced}
by noting that $G$ is $\vec N$-row reduced.
\end{proof}

\begin{example}
  We apply \cref{thm:dualityMinbasis} to the problem of \cref{ex:simpade} with shifts $\vec N = (5, 3, 4, 5)$.
  We have
  \begin{align*}
    A &= 
      \left[\begin{array}{rrr}
      x^{4} + x^{2} + 1 & x^{4} + 1 & x^{4} + x^{3} + 1 \\
      1 &   &   \\
        & 1 &   \\
        &   & 1
      \end{array}\right]
    \\
    B &= \left[\begin{array}{r}
             1 \\ x^4 + x^2 + 1 \\ x^{4} + 1 \\ x^{4} + x^{3} + 1
         \end{array}\right]
  \end{align*}
  An $\vec N$-minimal approximant basis to order $d = 5$ of $B$ is
  \begin{align*}
    G &= 
  \left[\begin{array}{rrrr}
  x & 0 & x & 0 \\
  1 & x^{2} + 1 & 0 & 0 \\
  0 & 1 & x^{2} + 1 & 0 \\
  0 & x & x + 1 & 1
  \end{array}\right] ,  \textrm{ and}
  \\
  \adj(G)\Transp &= 
    \left[\begin{array}{rrrr}
    x^{4} + 1 & x^{2} + 1 & 1 & x^{3} + 1 \\
    x & x^{3} + x & x & x^{4} + x \\
    x^{3} + x & x & x^{3} + x & x^{4} + x^{3} + x \\
    0 & 0 & 0 & x^{5}
    \end{array}\right]
    \ .
  \end{align*}
  $\adj(G)\Transp$ can be confirmed to be an $(-\vec N)$-minimal approximant basis of
$A$, since $\adj(G)\Transp A \equiv 0 \mod x^d$, and since the $(-\vec N)$-leading coefficient matrix of $\adj(G)\Transp$ has full rank.
\end{example}

Algorithm~\ref{alg:simpadedual} uses \cref{thm:dualityMinbasis} to solve a Simultaneous \Pade approximation by computing a minimal approximant basis of $B$ in Popov form.

\begin{algorithm}[t]
  \caption{\algoname{DualitySimPade}}
  \label{alg:simpadedual}
  \begin{algorithmic}[1]
    \Require{$(\vec S, (x^d,\ldots, x^d), \vec N)$, an instance of 
 \cref{prob:sim_pade_basis} of size~$n$.
    }
    \Ensure{$(\vec \lambda, \vec \delta)$, solution specification.}
    \State $B \ass [ 1, S_1, \ldots, S_n ]^T \in \K[x]^{(n+1) \times 1}$
    \State $G \ass \PopovBasis(d, B, \vec N)$
      \label{line:simpadedual:basis}
    \State $\vec \eta \ass \rowdeg_{\vec N} G$
    \State $\hat{\vec \lambda} \ass $ first column of $\adj(G\Transp)$
      \label{line:simpadedual:firstcol}
    \State $\hat{ \vec\delta} \ass (\eta - N -\eta_1, \ldots, \eta - N - \eta_{n+1})$,
    where $\eta = \sum_i \eta_i$  and $N = \sum_i N_i$
      \label{line:simpadedual:degrees}
    \State $I \ass \{ i \mid \hat{\vec\delta}_i < 0 \}$, and $k \ass |I|$
    \State $(\vec \lambda,\ \vec \delta) \ass \big( \hat{\vec\lambda}_{i \in I},\ (\hat{\vec\delta}_i)_{i \in I} 
\big) \in \K[x]^{k \times 1} \times \ZZ^{k} $
    \State $\return (\vec \lambda, \vec \delta)$
  \end{algorithmic}
\end{algorithm}

\begin{theorem} \cref{alg:simpadedual} is correct.
The cost of the algorithm is $O(n^{\omega-1}\, \M(d) (\log d) (\log d/n)^2)$
operations in $\K$.
\end{theorem}
\begin{proof}
Correctness follows from \cref{thm:dualityMinbasis}.  The complexity
estimate is achieved if the algorithms supporting \cref{thm:orderbasis}
and \cref{thm:fastsolver} are used for the computation in lines 2
and 4, respectively.
\end{proof}

\section{A Divide \& Conquer algorithm}
\label{sec:intersect}

Our second algorithm can handle the full generality of \cref{prob:sim_pade_basis}.
It works by first solving $n$ single \Pade approximations, one for each of the $S_i$ individually, and then intersecting these solutions to form approximations of multiple $S_i$ simultaneously.
The intersection is structured in a Divide \& Conquer tree, and performed by computing minimal approximant bases.
Let $(\vec S, \vec g, \vec N)$ be an instance of \cref{prob:sim_pade_basis}
of size $n$.

The idea of the intersection algorithm is the following:
consider that we have solution specifications for two different Simultaneous \Pade problems, $(\vec\lambda_1, \vec\delta_1)$ and $(\vec\lambda_2, \vec\delta_2)$.
We then compute an approximant basis $G$ of the following matrix:
\begin{equation}
  \label{eqn:intersect_R}
    R =
    \left[\begin{array}{@{}c|c@{}}
            1 & 1 \\ \hline
            -\vec\lambda_1 &        \\\hline
                   & -\vec\lambda_2 \\
    \end{array}\right]
\end{equation}
$G$ then encodes the \emph{intersection} of the $\K[x]$-linear combinations of the $\vec\lambda_1$ with the $\K[x]$-linear combinations of the $\vec\lambda_2$:
any $\lambda \in \K[x]$ residing in both sets of polynomials will appear as the first entry of a vector in the row space of $G$.
We compute $G$ as an $\vec r$-minimal approximant basis to high enough order, where $\vec r$ is selected carefully such that the $\vec r$-degree of any $(\lambda \mid \ldots) \in \Row(G)$ will equal the $(-\vec N)$-degree of the completion of $\lambda$ according to the combined Simultaneous \Pade problem, whenever this degree is negative.
From those rows of $G$ with negative $\vec r$-degree we then get a solution specification for the combined problem.

\begin{example}
  Consider again \cref{ex:simpade}.
  We divide the problem into two sub-problems $\vec S_1 = (S_1, S_2)$, $\vec N_1 = (5, 3, 4)$, and $\vec S_2 = (S_3)$ and $\vec N_2 = (5, 5)$.
  Note that $N_{1,0} = N_{2,0} = 5$, since this is the degree bound on the sought $\lambda$ for the combined problem.
  The sub-problems have the following solution specifications and their completions:
  \begin{align*}
    (\vec\lambda_1, \vec \delta_1) &= \big( [ x^{4} + 1,\ x^3 + x ]\Transp,\ ( -1, -1 ) \big)
    \\
      A_1 &=
        \left(\begin{array}{rrr}
        x^{4} + 1 & x^{2} + 1 & 1 \\
        x^{3} + x & x & x^{3} + x
        \end{array}\right)
    \\
    (\vec\lambda_2, \vec \delta_2) &= \big( [ x^2,\ x^3 + x + 1 ]\Transp,\ ( -3, -2 ) \big)
    \\
      A_2 &=
        \left(\begin{array}{rr}
        x^{2} & x^{2} \\
        x^{3} + x + 1 & x + 1
        \end{array}\right)
  \end{align*}
  We construct $R$ as in \eqref{eqn:intersect_R}, and compute $G$, a minimal approximant basis of $R$ of order $7$ and with shifts $\vec r= (-5 \mid -1, -1 \mid -3, -2)$ (the $G$ below is actually in $\vec r$-Popov form):
  \[
    G =  \left(\begin{array}{rrrrr}
    x^{8} & 0 & 0 & 0 & 0 \\
    x^{3} + x + 1 & x^{4} + 1 & 1 & 0 & 1 \\
    x^{3} + x^{2} + x + 1 & 1 & x + 1 & 1 & 1 \\
    x^{4} + x^{3} + x + 1 & 1 & 1 & x^{2} & 1 \\
    x^{4} + 1 & 1 & 0 & x + 1 & x + 1
    \end{array}\right)
  \]
  $G$ has $\vec r$-row degree $(3, 3, 0, -1, -1)$.
  Only rows 4 and 5 have negative $\vec r$-degree, and their first entries are the linearly independent solutions $x^4 + x^3 + x + 1$ and $x^4 + 1$.
  Both solutions complete into vectors with $(-\vec N)$-degree -1.
\end{example}

To prove the correctness of the above intuition, we will use \cref{alg:simpadedirect} (\algoname{DirectSimPade}).
The following lemma says that to solve two simultaneous \Pade approximations, one can compute a minimal approximant basis of one big matrix $A$ constructed essentially from two of the matrices employed in \algoname{DirectSimPade}.
Afterwards, \cref{lem:recursive_R_solves} uses this to show that a minimal approximant basis of $R$ in \eqref{eqn:intersect_R} provides the crucial information in a minimal approximant basis of $A$.

\begin{algorithm}[t]
  \caption{\algoname{RecursiveSimPade}}
  \label{alg:recsimpade}
  \begin{algorithmic}[1]
   \Require{$(\vec S, \vec g, \vec N)$, an instance of
   \cref{prob:sim_pade_basis} of size $n$.}
    \Ensure{$(\vec \lambda,  \vec \delta)$, a solution specification.}
    \If{$n=1$} 
     \State $\return \algoname{DirectSimPade}(\vec S, \vec g, \vec N)$
    \Else
      \State $\vec S_1, \vec g_1 \ass $
              the first $\ceil{n/2}$ elements of $\vec S, \vec g$
      \State $\vec S_2, \vec g_2 \ass $
              the last $\floor{n/2}$ elements of $\vec S, \vec g$
      \State $\vec N_1
         \ass (N_0,N_1,\ldots,N_{\ceil{n/2}})$
      \State $\vec N_2  
         \ass  (N_0,N_{\ceil{n/2}+1},\ldots,N_n)$
      \State $(\vec \lambda_1, \vec \delta_1) \ass \algoname{RecursiveSimPade}\big(\vec S_1, \vec g_1, \vec N_1)$
      \State $(\vec \lambda_2, \vec \delta_2) \ass \algoname{RecursiveSimPade}
   \big(\vec S_2, \vec g_2, \vec N_2)$
    \State $\vec r \ass (-N_0 \mid \vec \delta_1 \mid \vec \delta_2)$
    \State $d \ass N_0 + \max_i \deg g_i - 1$ \label{line:recursive:choosed}
    \State \label{lineH} $R \ass
    \left[\begin{array}{c|c}
            1 & 1\\ \hline
            -\vec \lambda_1 &  \\\hline
              & -\vec \lambda_2 
    \end{array}\right]$
    \State $(\left [\begin{array}{c|c} \vec \lambda & \ast \end{array} \right ],
 \vec \delta) \ass \NegMinBasis(d, R, \vec r)$
\label{calltoNegMin}
    \State $\return (\vec \lambda, \vec \delta)$
    \EndIf
  \end{algorithmic}
\end{algorithm}

\begin{lemma}
  \label{lem:recursive_big_matrix}
  Let $(\vec S_1, \vec g_1,\vec N_1)$  and $(\vec S_2, \vec g_2,\vec N_2)$ be two instances of \cref{prob:sim_pade_basis} of lengths $n_1, n_2$ respectively, and where $\vec N_1 = (N_0 \mid \grave{\vec N_1})$ and $\vec N_2 = (N_0 \mid \grave{\vec N_2})$.
  Let $\vec S = (\vec S_1 \mid \vec S_2)$, $\vec g = (\vec g_1 \mid \vec g_2)$ and $\vec N = (N_0 \mid \grave{\vec N_1} \mid \grave{\vec N_2})$ be the combined problem having length $n = n_1 + n_2$.

  Let $\vec h_i = (-\vec N_i \mid N_0 -1 \ldots N_0 -1) \in \ZZ^{2n_i+1}$ for $i=1,2$.
  Let $(F, \vec\delta) = \NegMinBasis(d, A, \vec a)$, where $A$ of dimension $(2n+3) \times (n+2)$ is given as:
  \[
    A =  \left [ \begin{array}{c|c}
                   A_1          & A_2
                 \end{array} \right ]
      = \left [ \begin{array}{cc|cc}
                                &                 & 1 & 1  \\
                -\vec S_1       &                 & -1     \\
                I               &                 &        \\
                \diagg[1] &                 &        \\
                                & -\vec S_2       &   & -1 \\
                                & I               &        \\
                                & \diagg[2] & 
                \end{array} \right ] ,
  \]
  with $\vec a = (- N_0 \mid \vec h_1 \mid \vec h_2)$ and $d = N_0 + \max_i \deg g_i - 1$.
  Then $(\vec \lambda, \vec \delta)$ is a solution specification to $(\vec S, \vec g, \vec N)$, where $\vec \lambda$ is the first column of $F$.
\end{lemma}
\begin{proof}
  Note that the matrix $A$ is right equivalent to the following matrix $B$:
  \[
  B := A 
      \left [ \begin{array}{cccc} 
           &    & I               &           \\
           &    &                 & I         \\
        1  &    & \vec S_1        &           \\
           & 1  &                 & \vec S_2
      \end{array} \right ] = 
    \left [ \begin{array}{cc|cc}
       1   & 1  & - \vec S_1      & -\vec S_2 \\
        -1 &    &                 &           \\
           &    & I               &           \\
           &    & \diagg[1] &           \\
           & -1 &                 &           \\
           &    &                 & I         \\
           &    &                 & \diagg[2]
    \end{array} \right ].
  \]
  Since $F$ is an ${\vec a}$-minimal approximant of $A$ of order $d$, then it will also be one for $B$.
  Let $P$ be the permutation matrix that produces the following 
  matrix $C := P B$:
  \[
\setlength{\arraycolsep}{.8\arraycolsep}
  C = P B = 
    \left [ \begin{array}{cc|cc}
    1  & 1  & - \vec S_1      & -\vec S_2       \\
       &    & I               &                 \\ 
       &    &                 & I               \\
       &    & \diagg[1] &                 \\
       &    &                 & \diagg[2] \\\hline
    -1 &    &                 &                 \\
       & -1 &                 & 
    \end{array}  \right ]  =
    \left [ \begin{array}{cc|c}
    1  & 1  & - \vec S                          \\     &  & I \\ &  & \diagg
                                                \\\hline  
    -1 &    &                                   \\
       & -1 & 
    \end{array}  \right ].
  \]
  Define $\vec c := \vec a P^{-1}$, and note that $\vec c = (\vec h \mid -N_0, -N_0)$.  
  Since $F = \NegMinBasis(d, A, \vec a)$, then $(FP^{-1}, \vec \delta)$ is a valid output of $\NegMinBasis(d, C, \vec c)$.
  Furthermore, since the first column of $P$ is $(1, 0, \ldots, 0)$, the first column of $F$ will be equal to the first column of $FP^\mo$.

  We are therefore finished if we can show that if $(F', \vec\delta')$ is any valid output of $\NegMinBasis(d, C, \vec c)$, then the first column of $F'$ together with $\vec \delta'$ form a solution specification to $(\vec S, \vec g, \vec N)$.

  \jsrn{I subtly changed the proof: we can't appeal to \cref{lem:paderec} because we need to prove the above statement for any possible output of $\NegMinBasis$, and not just the ones that decompose according to \cref{lem:paderec}.}

  Consider therefore such an $(F', \vec \delta')$.
  By the first two columns of $C$, we must have $F'_{*,1} \equiv F'_{*,2n+2} \equiv F'_{*,2n+3} \mod x^d$, where $F'_{*,i}$ denotes the $i$'th column of $F'$.
  Since each row of $F'$ have negative $\vec c$-degree, and since $N_0 < d$, then the congruences must lift to equalities.
  We can therefore write $F = [ G \mid F'_{*,1} \mid F'_{*,1} ]$ for some $G \in \K[x]^{k \times (2n+1)}$ for some $k$, and we have $\rowdeg_{\vec h} G = \rowdeg_{\vec c} F' = \vec \delta'$.

  By the last $n$ columns of $C$, we have $G H \equiv 0 \mod x^d$, where
  \[
    H =
    \left[\begin{array}{c}
            -\vec S \\
            I \\
            \diagg
          \end{array}\right] \ .
  \]
  In fact, $(G, \vec \delta')$ is a valid output for $\NegMinBasis(d, H, \vec h)$: for $G$ has full row rank since $F'$ does; $G$ is $\vec h$-row reduced since $F'$ is $\vec c$-row reduced; and any negative $\vec h$-order $d$ approximant of $H$ must clearly be in the span of $G$ since $F'$ is a negative $\vec c$-minimal approximant basis of $C$.

  By the choice of $d$, then \cref{thm:simpadedirect} therefore implies that the first column of $G$ together with $\vec \delta'$ form a solution specification to the problem $(\vec S, \vec g, \vec N)$.
  Since the first column of $G$ is also the first column of $F'$, this finishes the proof.
\end{proof}

\begin{lemma}
  \label{lem:recursive_R_solves}
  In the context of \cref{lem:recursive_big_matrix}, let $(\vec \lambda_1, \vec \delta_1)$ and  $(\vec \lambda_2, \vec \delta_2)$  be solution specifications to the two sub-problems, and let $\vec r = (-N_0 \mid \vec \delta_1 \mid \vec \delta_2)$.
  If $([ \vec\lambda \mid * ], \vec\delta) = \NegMinBasis(d, R, \vec r)$, where $\vec \lambda$ is a column vector and
  \[
    R = 
    \left[\begin{array}{c|c}
            1 & 1\\ \hline
            -\vec \lambda_1 &  \\\hline
              & -\vec \lambda_2 
          \end{array}\right] \ ,
  \]
  then $(\vec\lambda, \vec\delta)$ is a solution specification for the combined problem.
\end{lemma}
\begin{proof}
  We will prove the lemma by using \cref{lem:paderecprune} to relate valid outputs of $\NegMinBasis(d, R, \vec r)$ with valid outputs of $\NegMinBasis(d, A, \vec a)$ from \cref{lem:recursive_big_matrix}.

  For $i=1,2$, since $(\vec \lambda_i, \vec \delta_i)$ is a solution specification to the $i$'th problem, then by \cref{thm:simpadedirect} there is some $G_i \in \K[x]^{k_i \times 2n_i+1}$ whose first column is $\vec\lambda_i$ and such that $G_i$ is a valid output of $\NegMinBasis(d, H_i, \vec h_i)$, where
  \[
  H_i = \left [ \begin{array}{c}
                  -\vec S_i \\
                  I \\
                  \diagg[i]
                \end{array} \right ] \in \K[x]^{(2n_i+1) \times n_i} ,
  \]
  and $\vec h_i$ is as in \cref{lem:recursive_big_matrix}.
  Note now that if
  \[
    F_1 := \left [\begin{array}{ccc}
                    1 & &   \\
                      & G_1 \\
                      & & G_2
                  \end{array} \right ] \in \K[x]^{(k_1+k_2+1) \times (2n_1 + 2n_2 + 3)} ,
  \]
  then $(F_1, \vec r)$ is a valid output of $\NegMinBasis(d, A_1, \vec a)$: for $\rowdeg_{\vec a} F_1$ is clearly $\vec r$; $F_1$ has full row rank and is $\vec r$-row reduced; and the rows of $F_1$ must span all $\vec a$-order $d$ approximants of $A_1$, since the three column ``parts'' of $F_1$ correspond to the three row parts of $A_1$. \jsrn{Too informal?}.

  Note now that $F_1 A_2 = R$.
  Thus by \cref{lem:paderecprune}, if $(F_2, \vec \delta) = \NegMinBasis(d, R, \vec r)$, then $(F_2 F_1, \vec \delta)$ is a valid output of $\NegMinBasis(d, A, \vec a)$.
  Note that by the shape of $F_1$ then the first column $\vec\lambda$ of $F_2 F_1$ is the first column of $F_2$.
  Thus $\vec\lambda, \vec\delta$ are exactly as stated in the lemma, and by \cref{lem:recursive_big_matrix} they must be a solution specification to the combined problem.
\end{proof}
\begin{theorem}
  \cref{alg:recsimpade} is correct.  The cost of
the algorithm is 
$O(n^{\omega-1}\, \M(d) (\log d) (\log d/n)^2)$,
$d = \max_i \deg g_i$.
\end{theorem}
\begin{proof}
  Correctness follows from \cref{lem:recursive_R_solves}.
  For complexity, note that the choice of order in \cref{line:recursive:choosed} is bounded by $2\max_i \deg g_i$, i.e. twice the value of $d$ of this theorem.
  So if $T(n)$ is the cost \cref{alg:recsimpade} for given $n$ and where the order will be bounded by $O(d)$, then we have the following recursion:
  \[
    T(n) = \left \{\begin{array}{ll}
                     2T(n/2) + P(n) & \textrm{if } n > 1 \\
                     O(\M(d)\log d) & \textrm{if } n = 1 \textrm{ (see \cref{sec:direct_reduced_basis})}
                   \end{array}\right . \ ,
  \]
  where $P(n)$ is the cost of line \ref{calltoNegMin}.
  Using algorithm \algoname{PopovBasis} for the computation
  of the negative part of the minimal approximant bases
  we can set $P(n)$ to the target cost.
  The recursion then implies $T(n) \in O(P(n))$.
\end{proof}

\jsrn{For final version: discuss the fact that the algorithm is faster when there's only few dimensions of solutions to the individual problems, i.e.~a speed similar to in \cite{olesh_vector_2006}.}

%


\jsrn{Why do we present two algorithms?}

\noindent
{\bf Acknowledgements.}
The authors would like to thank George Labahn for valuable
discussions, and for making us aware of the Hermite--Simultaneous
\Pade duality.  We would also like to thank Vincent Neiger for
making preprints of \cite{jeannerod_computation_2016} available to
us.  The first author would like to thank the Digiteo Foundation
for funding the research visit at Waterloo, during which most of
the ideas of this paper were developed.

\bibliographystyle{abbrv}
\bibliography{bibtex}


\end{document}